\documentclass[11pt]{article}

\usepackage{array}
\usepackage{paralist}
\usepackage{graphicx}
\usepackage{amsmath}
\usepackage{amssymb}
\usepackage{amsthm}
\usepackage{mathtools}
\usepackage{fullpage}
\usepackage{hyperref}
\usepackage{tikz}

\newtheorem{thm}{Theorem}[section]
\newtheorem{claim}[thm]{Claim}
\newtheorem{lem}[thm]{Lemma}
\newtheorem{define}[thm]{Definition}
\newtheorem{cor}[thm]{Corollary}

\renewenvironment{proof}[1][]{
  \begin{trivlist}
   \item[\hspace{\labelsep}{\em\noindent Proof#1:\/}]}
   {{\hfill$\Box$}
  \end{trivlist}}

\DeclareMathOperator*{\E}{\mathbb{E}}
\DeclareMathOperator*{\dir}{dir}
\DeclareMathOperator*{\lev}{lev}
\DeclareMathOperator*{\edg}{Edge}
\DeclareMathOperator*{\spn}{span}
\renewcommand{\vec}[1]{\boldsymbol{#1}}

\DeclareMathOperator{\tr}{\mathsf{tr}}
\DeclareMathOperator{\rank}{\mathsf{rank}}
\DeclareMathOperator{\diag}{\mathsf{diag}}

\newcommand{\ip}[2]{\langle #1,#2 \rangle}
\def\span{\textsf{span}}
\def\weight{\textsf{weight}}

\newcommand{\R}{\mathbb{R}} 
\newcommand{\C}{\mathbb{C}} 
\newcommand{\N}{\mathbb{N}} 
\newcommand{\Z}{\mathbb{Z}} 
\newcommand{\F}{\mathbb{F}} 

\theoremstyle{remark}

\usepackage{xcolor}

\title{\bf Lower Bounds for Approximate LDCs}
\author{
Jop Bri{\"e}t\thanks{Courant Institute of Mathematical Sciences, New York University. 
Email: \texttt{jop.briet@cims.nyu.edu}.
Supported by a Rubicon grant from the Netherlands Organisation for Scientific Research (NWO).} 
\and
Zeev Dvir\thanks{Department of Computer Science and Department of Mathematics, Princeton University, Princeton NJ.
Email: \texttt{zeev.dvir@gmail.com}. Research partially
supported by NSF grants CCF-0832797, CCF-1217416 and by the Sloan fellowship.}
\and 
Guangda Hu\thanks{Department of Computer Science  Princeton University, Princeton NJ.
Email: \texttt{guangdah@cs.princeton.edu}. Research partially
supported by NSF grants CCF-0832797, CCF-1217416 and by the Sloan fellowship.}
\and
Shubhangi Saraf\thanks{Department of Computer Science and Department of Mathematics, Rutgers University.
 Email: \texttt{shubhangi.saraf@gmail.com}.}
}
\date{}

\begin{document}
\maketitle

\begin{abstract}
We study an approximate version of $q$-query LDCs (Locally Decodable Codes) over the real numbers and prove lower bounds on the encoding length of such codes. A $q$-query $(\alpha,\delta)$-approximate LDC is a set $V$ of $n$ points in  $\R^d$ so that, for each $i \in [d]$ there are $\Omega(\delta n)$ disjoint $q$-tuples  $(\vec{u}_1,\ldots,\vec{u}_q) $  in $V$ so that $\span(\vec{u}_1,\ldots,\vec{u}_q)$ contains a unit vector whose $i$'th coordinate is at least $\alpha$. We prove exponential lower bounds of the form $n \geq 2^{\Omega(\alpha \delta \sqrt{d})}$ for the case $q=2$ and, in some cases, stronger bounds (exponential in $d$). 
\end{abstract}

\section{Introduction}

Error Correcting Codes (ECCs) have always played an important part in the development of theoretical computer science. In particular, many of the foundational results of computational complexity rely in some way or another on constructions and analysis of ECCs (e.g., hardness of approximation, hardness-randomness tradeoffs). The study of ECCs from the perspective of complexity theorists sometimes has different a focus than the traditional information theory viewpoint. One such difference is the study of special kinds of codes that are useful for theory  (i.e., for proving theorems such as the PCP theorem) but were not studied previously.

One such example are Locally-Decodable-Codes (LDCs) which were formally defined in the seminal work of Katz and Trevisan \cite{KT00} (but were implicit in several prior works \cite{Blum-Kannan,Lipton90,BeaverF90}). These are codes that allow the receiver of a (possibly corrupted) encoding $y = C(x) \in \{0,1\}^n$ of a message $x \in \{0,1\}^d$ to probabilistically decode w.h.p a single message bit $x_i$ by reading only $q$ positions in $y$ (which might contain at most $\delta n$ errors). We usually think of $q$ as either a small constant or a very slow growing function of $n$ and of $\delta$ as a constant.
 
 The only case of LDCs which is mostly well understood is that of $2$-query codes (it is easy to see that $1$-query codes do not exist). The Hadamard code $C(x) = (\ip{x}{a})_{a \in \{0,1\}^d}$ is a $2$-query code with exponential encoding length. In \cite{GKST:2006, KdW04} it was shown that this is tight, that is, we always have $n \geq 2^{\Omega(\delta d)}$ for $2$-query codes. For $q>2$ there are huge gaps between the known lower and upper bounds. The best known lower bound is $n=\tilde{\Omega}(d^{1+1/(\lceil q/2\rceil-1)})$ for $q>4$ \cite{Woo07} and $n=\Omega(d^2)$ for $k=3,4$ \cite{KdW04,Woo10}. The best constructions for $q>2$ are given by Matching-Vector codes, which were introduced by Yekhanin in \cite{Yek08} and further developed in
\cite{Efr09,Rag07,KY09,IS10,CFL+10,DGY11,BET10}. These codes have block-length of roughly $n \leq \exp\exp\left((\log d)^{O(\log\log q/\log q)}(\log\log d)\right).$

One important sub-case of LDCs is that of {\em linear codes} (all known constructions are linear as far as we know). That is, the encoding is a linear mapping $C: \F^d \mapsto \F^n$ over some field $\F$. In this case, one can show that w.l.o.g. the decoding is linear as well. More formally, if we let $\vec{v}_1,\ldots,\vec{v}_n \in \F^d$ be the  rows of the generating matrix of $C$ (so that $C(\vec{x})_i = \ip{\vec{x}}{\vec{v}_i}$) then we have that, for each $i \in [d]$ there must exist a matching $M_i$ of at least  $\Omega(\delta n)$ disjoint pairs  $\vec{v}_{j_1},\vec{v}_{j_2}$  that span $\vec{e}_i$ (the $i$'th standard basis vector). To locally decode $x_i$ one can simply pick a random pair in the matching~$M_i$ and calculate:
$$ x_i = \ip{\vec{x}}{\vec{e}_i} = a\ip{\vec{x}}{\vec{v}_{j_1}}+b\ip{\vec{x}}{\vec{v}_{j_2}}$$ for some field elements $a,b$ satisfying $a\vec{v}_{j_1}+b\vec{v}_{j_2}= \vec{e}_i$. In \cite{DS05} is was shown that the lower bound of \cite{GKST:2006} for binary linear codes can be extended to linear codes over any field and so, we know that the Hadamard code cannot be beaten even if we allow for a large alphabet. 

In this work we consider a new notion of linear LDCs  in which the underlying field is the real numbers and the decoding is `approximate'.  Building on the above characterization of linear codes, we will consider arrangements of points $\vec{v}_1,\ldots,\vec{v}_n \in \R^d$ in which, for every $i \in [d]$ there are many disjoint pairs that `almost span' $\vec{e_i}$ in some concrete way (we give exact definitions below). Overall, our results are negative and show that, even if we allow a very loose notion of approximation, the encoding length is still exponential (either in $\sqrt{d}$ or in $d$, depending on the model). We prove several theorems for various settings of the parameters, using a wide  variety of techniques. 

\paragraph{Motivation and related works:} Our motivation for studying this problem comes from several directions. Firstly, one could hope to use approximate codes in practice (if these had sufficiently good parameters). As long as the approximation parameter is not too large we could hope to recover some approximation of $x_i$ using the two queries to the code (assuming $x_i$ is some quantity we are interested in and we don't mind some small error). Another motivation comes from trying to understand $3$-query codes. Here, even if we restrict our attention to real codes over $\R$, there is still an exponential gap between lower and upper bounds. In a recent work, \cite{DSW13}, a subset of the current authors and Avi Wigderson proved an $n > d^{2+\epsilon}$ lower bound (for some positive $\epsilon$) for a closely related notion of $2$-query Locally Correctable Codes (LCCs) over $\R$, improving upon the known quadratic bound. Originally, the proof of \cite{DSW13} used  a reduction from (exact) 3-LCCs over $\R$ to $2$-query approximate LDCs (later, a different proof was found). This raises the possibility that, in the future, perhaps approximate codes will find more applications. We are also motivated by connections to well studied questions in combinatorial geometry. In \cite{BDWY12,DSW12} it was shown that proving lower bounds on LCCs is closely related to questions in the spirit of the Sylvester-Gallai theorem. Here, one tries to take local information about a point configuration (say, many collinear triples)  and convert this information to a global bound on the dimension spanned by the points. We can view some of the theorems in this work in this spirit. Approximate versions of Sylveter-Gallai type theorems and LCCs were recently explored in~\cite{ADSW12}.

\subsection{Definitions and results}

We begin with some notations. A {\em $q$-matching} $M$ in $[n]$ is defined to be a set of disjoint unordered $q$-tuples (i.e. disjoint subsets of size $q$) of $[n]$. We denote by $\vec{e}_i$ the $i$'th standard basis vector in $\R^d$. The standard inner product of two vectors $\vec x,\vec y \in \R^d$ is given by $\ip{\vec x}{\vec y}$ and the $\ell_2$ norm of $\vec x \in \R^d$ is $\|\vec x\|_2 = \sqrt{\ip{\vec x}{\vec x}}$. 

\begin{define}[$\weight_i$]
For a vector $\vec{u} \in \R^d$ we define $\weight_i(\vec{u}) = |\ip{\vec{u}}{\vec{e}_i}|/\|\vec{u}\|_2$ (i.e., the absolute value of the $i$'th coordinate of the normalized vector $\vec{u}/\|\vec{u}\|_2$). 
\end{define}

Clearly we have $\sum_{i\in [d]}\weight_i(\vec{u})^2 = 1$.
We now state our definition of approximate LDC. 

\begin{define}[Approximate LDC]
Let~$d,n,q$ be positive integers and~$\alpha,\delta\in[0,1]$ real numbers. A $q$-query $(\alpha,\delta)$-approximate LDC is a pair $(V,M)$ with
\begin{enumerate}
	\item $V=\{\vec{v}_1,\vec{v}_2,\ldots,\vec{v}_n\}$ a multiset of vectors in $\R^d$. The parameter $n$ is the {\em size} (or {\em block length}) of the code and the parameter $d$ is the {\em dimension} (or  {\em message length}) of the code.
	\item $M = (M_1,\ldots,M_d)$ with each $M_i$ being a $q$-matching in $[n]$ so that, if $\{j_1,\ldots,j_q\} \in M_i$, then there exists $\vec{u} \in \span\{ \vec{v}_{j_1},\ldots,\vec{v}_{j_q} \}$ with $\weight_i(\vec{u}) \geq \alpha$.
\end{enumerate}
The sizes of the matchings $M_i$ must satisfy
$|M_1|+|M_2|+\cdots+|M_d|\geq\delta dn$
and the parameter $\delta$ is called the {\em density} of the code\footnote{The traditional definition would ask for each $M_i$ to be of size at least $\delta n$ but our definition is more general, which makes our (negative) results stronger.}.
\end{define}

Our first theorem gives an exponential bound on the block length of approximate $2$-LDCs for any $\alpha>0$. Notice that the bound gets worse as $\alpha$ approaches $1/\sqrt{d}$, at which point we cannot expect any lower bound to hold (since a single vector $\vec{u}$ can have $\weight_i(\vec{u})\geq 1/\sqrt{d}$ for all $i \in [d]$).

\begin{thm} \label{thm:general}[General lower bound]
A $2$-query $(\alpha,\delta)$-approximate LDC of size~$n$ and dimension~$d$ must satisfy $n \geq 2^{\Omega(\alpha\delta\sqrt{d})}.$
\end{thm}

We could hope to replace the exponential dependence on $\sqrt{d}$ with an exponential dependence on $d$ (as is the case with exact $2$-LDCs). In fact, we conjecture that a general bound of the form $n \geq \exp(\delta\alpha^2 d)$ should hold (the quadratic dependence on $\alpha$ is necessary to avoid hitting the $\alpha = 1/\sqrt{d}$ barrier). Currently, we are only able to prove this conjecture when $\alpha $ is sufficiently close to $1$. This is stated in the next theorem.

\begin{thm} \label{thm:special}[Lower bound for large $\alpha$]
Let $\alpha_0 = \sqrt{1-1/(4\pi^2)}\approx0.987$. A $2$-query   $(\alpha,\delta)$-approximate LDC of size $n$, dimension $d$ and $\alpha>\alpha_0$ must satisfy $n \geq 2^{\Omega(\delta d)},$ where the hidden constant in the~$\Omega(\cdot)$ depends on $\alpha - \alpha_0$.
\end{thm}

There is another special case where we can get an exponential dependence on $d$ instead of $\sqrt{d}$. It is a natural  restriction of the general definition but it requires two new notions (that will be useful in their own right down the road). The first is that of a {\em simple} code (we will only care about $2$-query codes).

\begin{define}[Simple code]
Let  $(V,M)$ be a $2$-query $(\alpha,\delta)$-approximate LDC.  We say that $(V,M)$ is a {\em simple} code if, for every $i\in[d]$ and $\{j_1,j_2\}\in M_i$ we have $\weight_i(\vec{v}_{j_2} - \vec{v}_{j_1} ) \geq \alpha$.
\end{define}

In other words, a simple code is an arrangements of points in $\R^d$ so that, for any $i\in [d]$ there are $\approx \delta n$ (on average) disjoint pairs of points that `point' in a direction that has projection at least~$\alpha$ on the $i$'th axis. An example of such an arrangement is the boolean cube $\{0,1\}^d \subset \R^d$ (all zero/one vectors), where $M_i$ consists of all $n/2$ pairs that differ only in the $i$'th entry (so $\alpha=1$).

Another feature of the hypercube is that all the  distances between pairs in $M_1,\ldots,M_d$ are equal (they all equal one), motivating the following definition.

\begin{define}[$c$-bounded]
Let  $c\geq 1$ and let $(V,M)$ be a $2$-query $(\alpha,\delta)$-approximate LDC.  We say that $(V,M)$ is  {\em $c$-bounded} if, for every $i\in[d]$ and $\{j_1,j_2\}\in M_i$ we have $\|\vec{v}_{j_2} - \vec{v}_{j_1}\| \in [1,c].$
\end{define}

The fact that the hypercube is both $c$-bounded (with $c=1$) and simple motivates the study of structures that satisfy these two conditions. In particular, we ask whether there exists a point arrangement in $\R^d$ which is `roughly' like the hypercube but has far fewer than~$2^d$ points. Here, the notion of `roughly' is captured by allowing pairwise distance to be `close' to 1 and the differences between adjacent vertices to be only somewhat axis parallel. The following theorem shows that such configurations do not exist (that is, you cannot beat the hypercube by much).
 
\begin{thm} \label{thm:bounded}[Lower bound for simple $c$-bounded LDCs]
A $2$-query $c$-bounded simple $(\alpha,\delta)$-approximate LDC of size $n$ and dimension $d$ must satisfy $n \geq 2^{\Omega(\alpha^2\delta^2d/(\log c)^2)}.$
\end{thm}

Finally, we consider arbitrary $q$-query approximate codes and observe that the lower bound proof of \cite{KT00} can be made to work also for approximate LDCs (with some additional work). This gives the following theorem.

\begin{thm} \label{thm:qquery}
Let $q\geq 1$ be an integer constant. A $q$-query $(\alpha,\delta)$-approximate LDC of size $n$ and dimension $d$ must satisfy  $n \geq \Omega((\alpha^2\delta^{1/q}d)^{\frac{q}{q-1}}).$
\end{thm}

\subsection{Techniques}

We briefly outline the techniques that appear in the proofs of our theorems. 

\paragraph{Simple codes:} An important ingredient in the proofs of Theorem~\ref{thm:general} and \ref{thm:special} is a general reduction from any approximate $2$-LDC to a simple code. The reduction follows by first normalizing the lengths of all vectors and then observing that if some linear combination $a\vec{v}_{j_1} + b\vec{v}_{j_2}$ has large $\weight_i$ then either one of the vectors $\vec{v}_{j_1},\vec{v}_{j_2}$ has large $\weight_i$ or the coefficients $a,b$ are close to $1,-1$. We can thus throw away all pairs in the matching $M_i$ in which one of the vectors has large $\weight_i$ and get a simple code (we do not throw away too many pairs since each vector has only a few large coordinates). Since this reduction does not preserve $c$-boundedness, we can unfortunately not use it to argue that Theorem~\ref{thm:bounded} works for non-simple $c$-bounded codes.

\paragraph{Proof of general bound:} The proof of Theorem~\ref{thm:general} (for simple codes w.l.o.g.) is via a recursive partitioning argument. In each step we pick a random $i \in [d]$ and partition $V$ into two sets using a random shift of a hyperplane orthogonal to $\vec{e}_i$. We analyze the expected number of edges (pairs in some $M_i$) cut in this process and show that is bounded by $O(\sqrt{d}/\alpha) \cdot \min\{|S|,|\bar S|\}$ with $S,\bar S$ representing the two parts of the cut.  The same inequality holds also when partitioning any subset $V' \subset V$ and so we can proceed recursively and obtain a bound of $O((\sqrt{d}/\alpha)n\log_2 n)$ on the total number of edges. Since this number is at least $\delta d n$ the theorem follows. This proof is inspired by the one appearing in \cite{GKST:2006} for exact (simple) $2$-LDCs.

\paragraph{Proof of bound for large $\alpha$:} Here we rely on a recent work of \cite{KORW12} which gives a (randomized) tiling of $\R^d$ with cells that have volume 1 and surface area $O(\sqrt{d})$ (same as a sphere up to a constant). This result gives a randomized rounding algorithm that we can leverage towards `rounding' our approximate code to an exact code (very roughly speaking) when $\alpha $ is large. This step is then combined with a random partitioning argument as in the proof of Theorem~\ref{thm:general}.

\paragraph{Proof for simple  $c$-bounded codes:} For this setting we use the LDC to construct a function $F$ from $\R^d$ to the space of complex $n \times n$ matrices given by
$F(\vec{x})=\left(e^{-i\langle \vec{x}, \vec{v}_s - \vec{v}_t\rangle}\right)_{s,t=1}^n.$
The crux of the proof applies an inequality relating the trace norms of the  first level (matrix) Fourier-coefficients of a matrix-valued function to its average trace norm (see Lemma~\ref{lem:traceineq}). The crucial observation is that the norms of the first level Fourier coefficients of the above defined $F$ can be lower bounded using the LDC property. The result then follows by combining this with the trivial upper bound on the average norm of $F$. This proof loosely follows an argument of~\cite{Ben-Aroya:2008} used for binary (non linear) LDCs and is inspired by work of~\cite{Briet:2012} linking LDCs to geometry of Banach spaces. 

\vspace{5mm}
\noindent{\bf{Organization: }}
We describe our reduction from general to simple 2-query codes in Section~\ref{sec:simple}. 
In Section~\ref{sec:general} we prove the bound for general codes (Theorem~\ref{thm:general}). In Section~\ref{sec:largealpha} we prove the bound for~$\alpha$ close to 1  (Theorem~\ref{thm:special}). In Section~\ref{sec:cbounded} we prove the bound for $c$-bounded codes (Theorem~\ref{thm:bounded}). Finally, in Section~\ref{sec:qquery} we prove the bound for general $q$-query approximate codes (Theorem~\ref{thm:qquery}).

\vspace{5mm}
\noindent{\bf{Acknowledgments: }} The authors would like to thank Avi Wigderson for many helpful conversations.

\section{Simple codes}\label{sec:simple}

In this section we prove the following theorem showing that any $2$-query approximate LDC can be transformed into a simple code with similar parameters.

\begin{thm}\label{thm-reducesimple}
If there exists a $2$-query $(\alpha,\delta)$-approximate LDC of size $n$ dimension $d$, then, for any integer $k>1/\alpha^2$, there exists a simple $2$-query $(\alpha',\delta')$-approximate LDC of size $n'$ and dimension $d$, where $\alpha'\geq\sqrt{\alpha^2-1/k}$, $\delta'\geq\delta-k/d$ and $n'\leq2n$.
\end{thm}

The main idea behind the proof of this result is as follows. Suppose that we have a pair of unit vectors $\vec{u},\vec{w} \in \R^d$ with $\weight_i(a\vec{u} + b\vec{w}) \geq \alpha$. It will be convenient to think of $\alpha$ as being close to one (the proof will work for any $\alpha$). So, after normalizing the coefficients $a,b$ we have that the unit vector $v = a\vec{u} + b\vec{w}$ is close to $\vec{e}_i$. We separate into two cases. In the first case, both $\vec{u}$ and $\vec{w}$ are almost orthogonal to $\vec{e}_i$. In this case, we must have that $\vec{u}-\vec{w}$ `points' in the direction of $\vec{e}_i$ (see diagrams in the complete proof) and so we don't really need the coefficients $a,b$. In the other case, at least one of $\vec{u},\vec{w}$ have significant inner product with $\vec{e}_i$. Notice, however, that, for each fixed $\vec{u}$,  this can only happen with a small number of $\vec{e}_i$'s when $i \in [n]$. These `bad' pairs can be removed from the matchings without causing a big decrease in their average size.

It will be convenient to use the following corollary of Theorem~\ref{thm-reducesimple} in which we set $k = \lceil 2/\alpha^2 \rceil$.
\begin{cor}\label{cor-reducesimple}
Suppose $d \geq 6/\alpha^2 \delta$. If there exists a $2$-query $(\alpha,\delta)$-approximate LDC of size $n$ dimension $d$, then there exists a simple $2$-query approximate $(\alpha',\delta')$-LDC of size $n'$ and dimension~$d$, where $\alpha'\geq \alpha/\sqrt{2}$, $\delta'\geq\delta/2$ and $n'\leq 2n$.
\end{cor}

We now move on to the formal proof of Theorem~\ref{thm-reducesimple}.

\begin{proof}[ of Theorem~\ref{thm-reducesimple}]
We first modify the code in the following way. Let $k>\frac{1}{\alpha^2}$ be a fixed integer.
\begin{enumerate}
\item 
For every matching $M_i$, remove the pairs~$\{j_1,j_2\}$ in which the $i$'th entry of~$\vec v_{j_1}$ is among its $k-1$ largest (in absolute value) entries or if this is the case for~$\vec v_{j_2}$.
Since every~$\vec v \in V$ causes at most~$k$ pairs (from all $d$ matchings) that contain it to be removed, there are at most $kn$ pairs removed in this step altogether.
\item Normalize all vectors in $V$ so that $\|\vec{v}_j\|_2=1$ for all $j \in [n]$ and discard all zero vectors.
\item For every $\vec{v}_j\in V$, add $-\vec{v}_j$ to $V$. For every original tuple $\{\vec{v}_{j_1},\vec{v}_{j_2}\}$, we replace it with two tuples: either $\{\vec{v}_{j_1},\vec{v}_{j_2}\}$, $\{-\vec{v}_{j_1},-\vec{v}_{j_2}\}$ or $\{\vec{v}_{j_1},-\vec{v}_{j_2}\}$, $\{-\vec{v}_{j_1},\vec{v}_{j_2}\}$ (to be determined later).
\end{enumerate}
Let $(V',M' = (M_1',\ldots,M_d'))$ be the vectors and matchings obtained from the above procedure.

\begin{claim} \label{cla:smallcor}
After the first step, if $\{j_1,j_2\}\in M_i$ is not deleted,
then
\[|v_{j_1i}|\leq\frac{1}{\sqrt{k}}\|\vec{v}_{j_1}\|_2\text{\quad and\quad}|v_{j_2i}|\leq\frac{1}{\sqrt{k}}\|\vec{v}_{j_2}\|_2.\]
\end{claim}

\begin{proof}
We only consider $\vec{v}_{j_1}$. According to the first step, the $i$'th coordinate must not be among the maximum $k-1$ ones. If this coordinate has absolute value greater than $\|\vec{v}_{j_1}\|_2/\sqrt{k}$, then there are at least~$k$ coordinates greater than $\|\vec{v}_{j_1}\|_2/\sqrt{k}$, which is impossible.
\end{proof}

\begin{claim} \label{cla:2d}
After the first step, for any remaining pair $\{j_1,j_2\} \in M_i$, $\vec{v}_{j_1}$ and $\vec{v}_{j_2}$ are linearly independent. This implies that no remaining pair contains $\vec{0}$ and so discarding all zero vectors in step 2 above does not remove any additional pairs from the matchings.
\end{claim}
\begin{proof}
Assume the contrary and $\vec{v}_{j_1}\neq\vec{0}$. Then $\spn\{\vec{v}_{j_1},\vec{v}_{j_2}\}$ contains only multiples of~$\vec{v}_{j_1}$. Thus, the $i$'th coordinate of $\vec{v}_{j_1}$ has magnitude at least $\alpha\|\vec{v}_{j_1}\|_2>\|\vec{v}_{j_1}\|_2/\sqrt{k}$, violating Claim~\ref{cla:smallcor}.
\end{proof}

For the new code $(V',M')$ the dimension is $d'=d$, the size is $n'\leq2n$, and the number of tuples is at least $2(\delta dn-kn)$, which implies the density $\delta'\geq\delta-k/d$. Notice that we might have removed some of the matchings completely. We still, however, consider the dimension as $d'=d$ (since we only use the sum of sizes of all matchings).

Next we argue that the pairs remaining after step 1 above give a simple code (up to changing signs) and calculate the parameter $\alpha'$. Fix a pair $\{{j_1},{j_2}\}\in M_i$ that remains after the first step. We will show that either $\vec{v}_{j_1}-\vec{v}_{j_2}$ or $\vec{v}_{j_1}-(-\vec{v}_{j_2})$ has a large $i$'th coordinate. Precisely, we show that either
\[
\weight_i(\vec{v}_{j_1}-\vec{v}_{j_2}) \geq\sqrt{\alpha^2-\frac{1}{k}}\text{\quad or\quad} \weight_i(\vec{v}_{j_1}-(-\vec{v}_{j_2}))\geq\sqrt{\alpha^2-\frac{1}{k}}.\]
Then if the first one holds, we choose to use $\{\vec{v}_{j_1},\vec{v}_{j_2}\}$ and $\{-\vec{v}_{j_1},-\vec{v}_{j_2}\}$ at the third step; otherwise if the second one holds, we select $\{\vec{v}_{j_1},-\vec{v}_{j_2}\}$ and $\{-\vec{v}_{j_1},\vec{v}_{j_2}\}$. Thus the code is reduced to a simple code with $\alpha'\geq\sqrt{\alpha^2- 1/k}$.

We consider the plane determined by $\vec{v}_{j_1}$ and $\vec{v}_{j_2}$, and set up Cartesian axes at the origin. Let the projection of $\vec{e}_i$ on this plane be in the direction of the $y$ axis and choose one of two possible directions for the $x$ axis arbitrarily. Let $\tau$ be the angle between the plane and $\vec{e}_i$. This setting is shown in figure~(a).

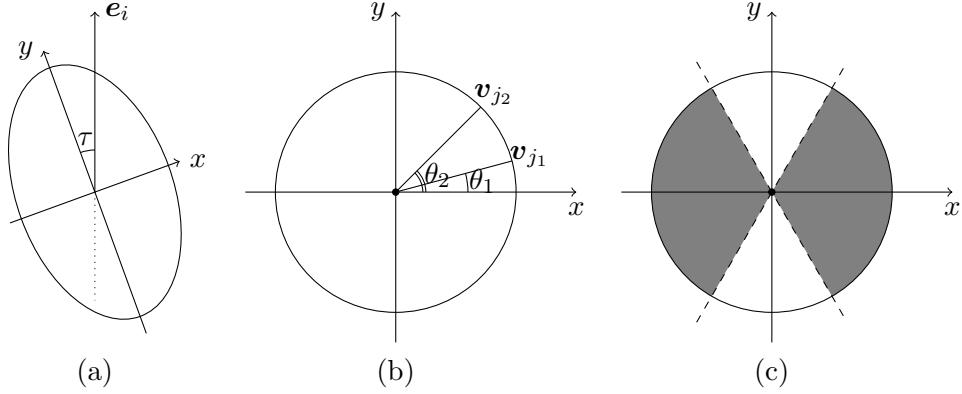
\begin{center}
\begin{tikzpicture}
\begin{scope}[xshift=-4cm, scale=0.8]
\draw[->] (0,0) -- (0,3) node[right] {$\vec{e}_i$};
\draw[dotted] (0,0) -- (0,-1.8);
\draw[rotate=110] (0,0) ellipse (2.2 and 1.3);
\draw[rotate=110,->] (-2.5,0) -- (2.5,0) node[left] {$y$};
\draw[rotate=110,->] (0,1.5) -- (0,-1.5) node[right] {$x$};
\draw (0,0.7) arc(90:110:0.7);
\node at (100:0.9) {$\tau$};
\node at (0,-3) {(a)};
\end{scope}
\begin{scope}[xshift=0, scale=0.8]
\draw[fill] (0,0) circle (1.5pt);
\draw (0,0) circle (2cm);
\draw[->] (-2.5,0) -- (3,0) node[below] {$x$};
\draw[->] (0,-2.5) -- (0,3) node[left] {$y$};

\draw (0,0) -- (15:2cm) node[pos=1.15] {$\vec{v}_{j_1}$};
\draw (0,0) -- (45:2cm) node[pos=1.15] {$\vec{v}_{j_2}$};
\draw (1.2,0) arc(0:15:1.2);
\node at (7.5:1.45) {$\theta_1$};
\draw (0.5,0) arc(0:45:0.5);
\draw (0.45,0) arc(0:45:0.45);
\node at (22.5:0.75) {$\theta_2$};
\node at (0,-3) {(b)};
\end{scope}
\begin{scope}[xshift=5cm, scale=0.8]
\fill[fill=gray] (0,0) -- (-60:2cm) arc(-60:60:2cm) -- cycle;
\fill[fill=gray] (0,0) -- (120:2cm) arc(120:240:2cm) -- cycle;

\draw[fill] (0,0) circle (1.5pt);
\draw (0,0) circle (2cm);
\draw[->] (-2.5,0) -- (3,0) node[below] {$x$};
\draw[->] (0,-2.5) -- (0,3) node[left] {$y$};

\draw[dashed] (240:2.5cm) -- (60:2.5cm);
\draw[dashed] (120:2.5cm) -- (300:2.5cm);
\node at (0,-3) {(c)};
\end{scope}
\end{tikzpicture}
\end{center}

Since $\{{j_1},{j_2}\}\in M_i$, we  see that
\begin{equation} \label{eq:coslb}
\cos\tau\geq\alpha.
\end{equation}

Let $\theta_1\in[0,2\pi)$ be the angle between $\vec{v}_{j_1}$ and the $x$ axis, and $\theta_2\in[0,2\pi)$ be the angle between~$\vec{v}_{j_2}$ and the $x$ axis. This is shown as in figure~(b). $\vec{v}_{j_1}$ and $\vec{v}_{j_2}$ correspond to points $(\cos\theta_1,\sin\theta_1)$ and $(\cos\theta_2,\sin\theta_2)$ on the plane.

By Claim~\ref{cla:smallcor}, the $i$'th coordinates of $\vec{v}_{j_1}$ and $\vec{v}_{j_2}$ are at most $1/\sqrt{k}$. Therefore
\[|\sin\theta_1|\cos\tau\leq\frac{1}{\sqrt{k}} \quad \text{and} \quad |\sin\theta_2|\cos\tau\leq\frac{1}{\sqrt{k}}.\]
We can see that the angles of $\pm\vec{v}_{j_1}$ and $\pm\vec{v}_{j_2}$, which are $\pm\theta_1$ and $\pm\theta_2$, must fall into two regions
\[\left[-\arcsin\frac{1}{\sqrt{k}\cos\tau},\arcsin\frac{1}{\sqrt{k}\cos\tau}\right]\text{ and }\left[\pi-\arcsin\frac{1}{\sqrt{k}\cos\tau},\pi+\arcsin\frac{1}{\sqrt{k}\cos\tau}\right].\]
These are shown as two gray circular sectors (the left one and the right one) in figure~(c).

We pair $\vec{v}_{j_1}$ to the one of $\pm\vec{v}_{j_2}$ that lies in the same sector with $\vec{v}_{j_1}$, and do the same for $-\vec{v}_{j_1}$. It is easy to see our pairing is either $\{\vec{v}_{j_1},\vec{v}_{j_2}\}$, $\{-\vec{v}_{j_1},-\vec{v}_{j_2}\}$ or $\{\vec{v}_{j_1},-\vec{v}_{j_2}\}$, $\{-\vec{v}_{j_1},\vec{v}_{j_2}\}$. We now argue that subtraction of vectors in a one of the new pairs (which belong to the same sector) must have a large $i$'th coordinate. Without loss of generality, we assume $\vec{v}_{j_1}$ and $\vec{v}_{j_2}$ are paired and only consider this pair.

The vector $\vec{v}_{j_2}-\vec{v}_{j_1}$ is parallel to the tangent line to the unit circle at angle $(\theta_1+\theta_2)/2$. One can verify that  a unit vector parallel to $\vec{v}_{j_2}-\vec{v}_{j_1}$ must have $y$ coordinate $\pm\cos(\theta_1+\theta_2)/2$. Therefore
\[\frac{|\langle\vec{v}_{j_1}-\vec{v}_{j_2},\vec{e}_i\rangle|}{\|\vec{v}_{j_1}-\vec{v}_{j_2}\|_2}=\cos\frac{\theta_1+\theta_2}{2}\cdot\cos\tau=\sqrt{1-\sin^2\frac{\theta_1+\theta_2}{2}}\cdot\cos\tau.\]
Since $\vec{v}_{j_1}$ and $\vec{v}_{j_2}$ are in the same circular sector,  $|\sin(\theta_1+\theta_2)/2|\leq1/(\sqrt{k}\cos\tau)$. It follows that
\[\frac{|\langle\vec{v}_{j_1}-\vec{v}_{j_2},\vec{e}_i\rangle|}{\|\vec{v}_{j_1}-\vec{v}_{j_2}\|_2}\geq\sqrt{\cos^2\tau-\frac{1}{k}}\geq\sqrt{\alpha^2-\frac{1}{k}}.\]
Here we used $\cos\tau\geq\alpha$ (Inequality~(\ref{eq:coslb})). This completes the proof.
\end{proof}

\section{Lower Bound for General Simple Codes} \label{sec:general}

We associate with a simple code $C = (V,M)$ a labeled graph $G_C$  on vertex set $V$ with edges given by all pairs in $M_1,\ldots,M_d$. We label each edge in $M_i$ with the label $i$ and allow for parallel edges (with different labels). We refer to the label of an edge $e$ as the {\em direction} of the edge and denote it by $\dir(e) \in [d]$. The proof will follow by analyzing cuts in the graph $G_C$, which we assume contains at least $\delta d n$ edges.

For $S\subseteq V$, let $\edg(S)$ be the set of edges of $G_C$ with both end points in $S$. We say that $(S_1,S_2)$ is a {\em cut} if $S_1\cup S_2=S$ and $S_1\cap S_2=\emptyset$. The cut is {\em non-trivial} if $S_1,S_2\neq\emptyset$. We use $\edg(S_1,S_2)$ to denote the set of edges with one endpoint in $S_1$ and the other in $S_2$.

The next lemma of~\cite[Appendix]{GKST:2006} relates the sizes of cuts in the graph with the total number of edges (the lemma holds for any graph). We include its proof for completeness.
\begin{lem} \label{lem:cut}
Suppose that for every $S\subseteq V$ with $|S|\geq2$, there exists a non-trivial cut $(S_1,S_2)$ satisfying $|\edg(S_1,S_2)|\leq c\cdot\min\{|S_1|,|S_2|\}$, then  $G_C$ has at most $ \frac{c}{2}|V|\log_2|V|$ edges.
\end{lem}

\begin{proof}
The given condition is equivalent to
\begin{equation}\label{eq:condition}
|\edg(S)|\leq c\cdot\min\{|S_1|,|S_2|\}+|\edg(S_1)|+|\edg(S_2)|.
\end{equation}
We induct on $|S|$ to show $\edg(S)\leq\frac{c}{2}|S|\log_2|S|$ for every non-empty $S\subseteq V$. For $|S|=1$, this is trivial. Assume this is true for $|S|<k$, $k\geq2$, and consider a subset $S$ of size $|S|=k$. 
Let $S_1\subseteq S$ be a proper and nonempty subset of $S$ and let~$S_2 = S\backslash S_1$ be its complement in~$S$.
Then, by the above condition~\ref{eq:condition} and the induction hypothesis,
\[
|\edg(S)|  \leq c\cdot\min\{|S_1|,|S_2|\}+\frac{c}{2}|S_1|\log_2|S_1|+\frac{c}{2}|S_2|\log_2|S_2|
\]
Assume (w.l.o.g) that $\min\{|S_1|,|S_2|\} = |S_1|$ and let $\eta = |S_1|/k$.
Notice that~$\eta$ belongs to~$[0,1/2]$.
Then the above right-hand side can be re-written as $(ck/2)\big(2\eta - H_2(\eta)\big) + (ck/2)\log k$, where $H_2(\tau) = -\tau\log_2 \tau - (1-\tau)\log_2 \tau$ is the binary entropy function.
The function $H_2$ is concave on $[0,1/2]$ and satisfies $H_2(0) = 0$ and $H_2(1/2) = 1$. 
Hence, the term $2\eta - H_2(\eta)$ is non-positive and we get the result $|\edg(S)| \leq (ck/2)\log k$, as claimed.
\end{proof}

We now proceed to prove Theorem~\ref{thm:general}. We will show $n=2^{\alpha\delta\sqrt{d}}$ for any $(\alpha,\delta)$ simple code (the general case will follow using Corollary~\ref{cor-reducesimple}). This will follow by combining the following lemma and Lemma~\ref{lem:cut}.

\begin{lem}
Let $C=(V,M)$ be an $(\alpha,\delta)$ simple code and let $G_C$ be the associated graph described above. Then, for any $S\subseteq V$ with $|S|\geq2$, there exists a non-trivial cut $(S_1,S_2)$ such that
\[|\edg(S_1,S_2)|\leq\frac{2\sqrt{d}}{\alpha}\cdot\min\{|S_1|,|S_2|\}.\]
\end{lem}

\begin{proof}
If $S$ contains no edges, an arbitrary cut will satisfy the requirement. We thus  assume that~$S$ contains at least one edge. We now analyze the size of a random cut chosen in a specific way.

Assume all points in~$V$ are in a ($d$-dimensional) box of edge length $L$. We pick a random direction $i\in[d]$ and then pick a plane perpendicular to $\vec{e}_i$ at a random position intersecting the box. The plane cuts the box into two parts. We define~$S_1$ to be the set of points in one part and~$S_2$ to be the set of points in the other part (the probability of having a point on the hyperplane is zero). We analyze the edges in this cut $(S_1,S_2)$.
We say that an edge $e\in\edg(S_1,S_2)$ is cut in the {\em right} direction if the plane is perpendicular to the direction of $e$, i.e. $\dir(e)=i$.

We consider a specific edge. Let $e_0=\{\vec{v}_{j_1},\vec{v}_{j_2}\}$ with $\{j_1,j_2\}\in M_{i_0}$ be an edge in direction $\dir(e_0)=i_0$ and denote $\vec{v}_{j_1}-\vec{v}_{j_2}=(u_1,u_2,\ldots,u_d)$.

For each $i'\in[d]$ the probability that $e_0$ is cut by a plane perpendicular to~$\vec{e}_{i'}$~is
\[\Pr[i=i']\cdot\Pr[\text{the plane falls between }v_{j_1i'}\text{ and }v_{j_2i'}]=\frac{1}{d}\cdot\frac{|u_{i'}|}{L}.\]

Therefore,
\[\Pr[e_0\in\edg(S_1,S_2)]=\sum_{i'=1}^d\frac{1}{d}\cdot\frac{|u_{i'}|}{L}=\frac{|u_1|+|u_2|+\cdots+|u_d|}{dL}.\]
Moreover, by the definition of an approximate code ($|u_{i_0}|\geq\alpha\|\vec{v}_{j_1}-\vec{v}_{j_2}\|_2$) and the Cauchy-Schwarz inequality, edge $e_0$ is cut in the right direction with probability
\begin{eqnarray*}
\frac{|u_{i_0}|}{dL} & \geq & \frac{1}{dL}\cdot\alpha\sqrt{u_1^2+u_2^2+\cdots+u_d^2} \\
& \geq & \frac{1}{dL}\cdot\frac{\alpha}{\sqrt{d}}\big(|u_1|+|u_2|+\cdots+|u_d|\big) \\
& = & \frac{\alpha}{\sqrt{d}}\Pr[e_0\in\edg(S_1,S_2)].
\end{eqnarray*}

Since $\vec{v}_{j_1}-\vec{v}_{j_2}$ has at least one non-zero coordinate, $\Pr[e_0\in\edg(S_1,S_2)]$ must be strictly positive. It follows that edge $e_0$ is cut in the right direction with probability strictly greater than
\[\frac{\alpha}{2\sqrt{d}}\Pr[e_0\in\edg(S_1,S_2)].\]
Hence, the expected number of edges that are cut in the right direction is strictly greater than $\alpha\E\big[|\edg(S_1,S_2)|\big]/(2\sqrt{d})$.
There must therefore exist an $i\in[d]$, a plane perpendicular to~$\vec{e}_i$ and a corresponding cut $(S_1,S_2)$  which cuts strictly more than $\alpha|\edg(S_1,S_2)|/(2\sqrt{d})$ in the right direction.
Since this number is  non-negative, there must be at least one edge cut in the right direction. This implies that $S_1$ and~$S_2$ are not empty.

All edges cut in the right direction must have the same direction $i$. Hence, these edges are disjoint (they form a matching in $V$), implying that the total number of cut edges is at most $\min\{|S_1|,|S_2|\}$. It follows immediately that
\[|\edg(S_1,S_2)|\leq\frac{2\sqrt{d}}{\alpha}\min\{|S_1|,|S_2|\}.\]
Therefore the cut $(S_1,S_2)$ satisfies the requirement.
\end{proof}

Now using Lemma~\ref{lem:cut} we can conclude that 
$ \delta d n \leq (2\sqrt{d}/\alpha) n \cdot \log_2 n, $ which gives
$n\geq 2^{\alpha\delta\sqrt{d}}$ as required. This completes the proof of Theorem~\ref{thm:general}.

\section{Lower Bound for Simple Codes with Large $\alpha$} \label{sec:largealpha}

In this section we prove Theorem~\ref{thm:bounded}. By Theorem~\ref{thm-reducesimple} it is enough to consider simple codes (the general case will follow by applying Theorem~\ref{thm-reducesimple} with $k$ a sufficiently large constant). We will use the definition and terminology of the graph $G_C$ defined in the last section for simple codes. Hence, we think of pairs in $M_i$ as edges in `direction' $\dir(e) =i$. We define the {\em length} of an edge $e = \{\vec{v}_{j_1},\vec{v}_{j_2}\}$ to be $\|\vec{v}_{j_1} - \vec{v}_{j_2}\|_2$.

We will use a recent  result of \cite{KORW12} concerning a partitioning (or tiling) of $\R^d$. Let $G=\{g\vec{z}\mid\vec{z}\in\Z^d\}$ be the set of grid points with grid distance $g\in\R^+$. Suppose we have a cell containing the origin and no other  points of $G$. We  can attempt to tile the space by taking all the shifts of this cell by all vectors in $G$. Clearly, one can do this using square tiles. However, it was an open problem to find the `most efficient' way of tiling $\R^d$ (in some well defined geometric sense of `efficient'). \cite{KORW12} gives a randomized algorithm outputting the shape of the cell so that the entire space is fully covered and no two cells overlap (thus, it is a tiling) and each cell corresponds to one grid point. Let $C(\vec{x})\in G$ ($\vec{x}\in\R^d$) denote the grid point in the cell containing $\vec{x}$ (so we can think of $C(\vec{x})$ as a `rounding' of $\vec{x}$). \cite{KORW12} proved the following\footnote{\cite{KORW12} only considered $g=1$ but the general result follows by simple scaling.}:

\begin{thm}[\cite{KORW12}]\label{thm-KROW}
There is a randomized algorithm partitioning the whole space $\R^d$ into cells such that
\begin{enumerate}
\item For every $\vec{x}\in\R^d$ and $\vec{s}\in G$, $C(\vec{x}+\vec{s})=C(\vec{x})+\vec{s}$.
\item For every two points $\vec{x},\vec{y}\in\R^d$, $\Pr[C(\vec{x})\neq C(\vec{y})]\leq2\pi\|\vec{y}-\vec{x}\|_2/g$.
\end{enumerate}
\end{thm}

Let $\epsilon\in(0,1)$ and $t\in\Z^+$ be two parameters to be determined later. We partition $\R^+$ into sets $\R^+=I_0\cup I_1\cup\cdots\cup I_{t-1}$, where
\[I_j=\bigcup_{k\in\Z}\left[(1+\epsilon)^{kt+j},(1+\epsilon)^{kt+j+1}\right).\]

For $I\subseteq\R^+$, we say an edge is {\em contained} in $I$ if its length falls in $I$. Without loss of generality we assume $I_0$ is the one among $\{I_0,I_1,\ldots,I_{t-1}\}$ that contains the most edges. We remove all edges not contained in $I_0$. The density $\delta$ is decreased by a factor of at most $t$.

Recall that $I_0=\bigcup_{k\in\Z}\left[(1+\epsilon)^{kt},(1+\epsilon)^{kt+1}\right)$. We say that the {\em level} of an edge is $k$ if it is contained in $\left[(1+\epsilon)^{kt},(1+\epsilon)^{kt+1}\right)$. For an edge $e$, we use $\lev(e)$ to denote its level. Let $k_{\min}$ and $k_{\max}$ to be the minimum level and the maximum level of all edges respectively.

For every integer $k\in[k_{\min},k_{\max}]$ we use Theorem~\ref{thm-KROW} to generate an (independent) random partition with grid distance
\[g_k=\frac{(1+\epsilon)^{kt}+(1+\epsilon)^{kt+1}}{2\alpha}=\frac{(2+\epsilon)(1+\epsilon)^{kt}}{2\alpha},\]
and use $C_k(\vec{x})$ to denote the corresponding rounding function.

Consider an edge $e=\{\vec{v}_{j_1},\vec{v}_{j_2}\}$ and say $\dir(e)=i_0$. We assume $\langle\vec{v}_{j_2}-\vec{v}_{j_1},\vec{e}_{i_0}\rangle>0$. (Otherwise we switch the order of $\vec{v}_{j_1}$ and $\vec{v}_{j_2}$.) We say the edge is {\em good} if the following properties are satisfied:
\begin{enumerate}
\item For $k=\lev(e)$, $C_k(\vec{v}_{j_1}+g_k\vec{e}_{i_0})=C_k(\vec{v}_{j_2})$. Since $C_k(\vec{v}_{j_1}+g_k\vec{e}_{i_0})=C_k(\vec{v}_{j_1})+g_k\vec{e}_{i_0}$, this means that the two cells containing $\vec{v}_{j_1}$ and $\vec{v}_{j_2}$ are  adjacent along the direction $\vec{e}_{i_0}$. 
\item For $k>\lev(e)$, $C_k(\vec{v}_{j_1})=C_k(\vec{v}_{j_2})$. In other words, the two ends are in the same cell.
\end{enumerate}

\begin{center}
\begin{tikzpicture}
\draw (0,0) -- (4,0);
\node at (5,-0.1) {$\vec{v}_{j_1}+g_k\vec{e}_{i_0}$};
\draw (0,0) -- (4,0.3);
\node at (4.3,0.4) {$\vec{v}_{j_2}$};
\node at (-0.5,0) {$\vec{v}_{j_1}$};
\draw[->] (7.5,0) -- (9,0);
\node at (8.25,0.25) {$\vec{e}_{i_0}$};
\end{tikzpicture}
\end{center}

\begin{claim}
Every edge is good with probability at least
\[1-\left(2\pi\sqrt{1-\alpha^2+\left(\frac{\alpha\epsilon}{2+\epsilon}\right)^2}+\frac{4\pi\alpha(1+\epsilon)}{(2+\epsilon)\left((1+\epsilon)^t-1\right)}\right).\]
\end{claim}

\begin{proof}
We consider the edge $e=\{\vec{v}_{j_1},\vec{v}_{j_2}\}$ with direction $i_0$, and assume $\langle\vec{v}_{j_2}-\vec{v}_{j_1},\vec{e}_{i_0}\rangle>0$. Then $\langle\vec{v}_{j_2}-\vec{v}_{j_1},\vec{e}_{i_0}\rangle\geq\alpha\|\vec{v}_{j_2}-\vec{v}_{j_1}\|_2$ and
\[\frac{\|\vec{v}_{j_2}-\vec{v}_{j_1}\|_2}{g_{\lev(e)}}\in\Big[\frac{2\alpha}{2+\epsilon},\frac{2\alpha(1+\epsilon)}{2+\epsilon}\Big)=\left[\alpha-\frac{\alpha\epsilon}{2+\epsilon},\alpha+\frac{\alpha\epsilon}{2+\epsilon}\right),\]

We consider the probability that $e$ is not a good edge.
\begin{enumerate}
\item For $k=\lev(e)$, we have
\begin{eqnarray*}
\Pr\left[C_k(\vec{v}_{j_1}+g_k\vec{e}_{i_0})\neq C_k(\vec{v}_{j_2})\right] & \leq & 2\pi\|\vec{v}_{j_2}-(\vec{v}_{j_1}+g_k\vec{e}_{i_0})\|_2/g_k \\
& \leq & 2\pi/g_k\cdot\sqrt{\|\vec{v}_{j_2}-\vec{v}_{j_1}\|_2^2+g_k^2-2\alpha g_k\|\vec{v}_{j_2}-\vec{v}_{j_1}\|_2} \\
& = & 2\pi\sqrt{\left(\|\vec{v}_{j_2}-\vec{v}_{j_1}\|_2/g_k-\alpha\right)^2+1-\alpha^2} \\
& \leq & 2\pi\sqrt{1-\alpha^2+\left(\frac{\alpha\epsilon}{2+\epsilon}\right)^2}.
\end{eqnarray*}
\item For $k>\lev(e)$, we have
\begin{eqnarray*}
\Pr\left[C_k(\vec{v}_{j_1})\neq C_k(\vec{v}_{j_2})\right] & \leq & 2\pi\|\vec{v}_{j_2}-\vec{v}_{j_1}\|_2/g_k \\
& \leq & 2\pi\cdot\frac{2\alpha(1+\epsilon)}{2+\epsilon}\cdot\frac{g_{\lev(e)}}{g_k} \\
& = & \frac{4\pi\alpha(1+\epsilon)}{2+\epsilon}\cdot\frac{1}{(1+\epsilon)^{(k-\lev(e))t}}.
\end{eqnarray*}
\end{enumerate}
By union bound, the probability that $e$ is not a good edge is at most
\begin{eqnarray*}
& & 2\pi\sqrt{1-\alpha^2+\left(\frac{\alpha\epsilon}{2+\epsilon}\right)^2}+\sum_{k=\lev(e)+1}^{k_{\max}}\left(\frac{4\pi\alpha(1+\epsilon)}{2+\epsilon}\cdot\frac{1}{(1+\epsilon)^{(k-\lev(e))t}}\right) \\
& < & 2\pi\sqrt{1-\alpha^2+\left(\frac{\alpha\epsilon}{2+\epsilon}\right)^2}+\frac{4\pi\alpha(1+\epsilon)}{(2+\epsilon)\left((1+\epsilon)^t-1\right)}.
\end{eqnarray*}
Thus the claim is proved.
\end{proof}

For any $\alpha>\sqrt{1-1/(4\pi^2)}$, we can always pick $\epsilon$ sufficiently small and $t$ sufficiently large so that each edge is good with positive probability. For example, if $\alpha=0.99$, we can take $\epsilon=0.01$ and $t=500$, in which case each edge is good with probability at least $0.069$. For simplicity, we use $O(\cdot)$ and $\Omega(\cdot)$ to suppress the exact values of constants $\alpha$, $\epsilon$ and $t$. The above claim tells us every edge is good with probability $\Omega(1)$. By a simple expectation argument, there exists a series of space partitions (for every $k\in[k_{\min},k_{\max}]$) such that $\Omega(1)$ fraction of all edges are good. We fix these partitions and remove all edges that are not good. In the remaining code the density is $\Omega(\delta)$.

Next, we prove the lower bound $n=2^{\Omega(\delta d)}$. This follows immediately from the following lemma and Lemma~\ref{lem:cut}.

\begin{lem}
For any $S\subseteq V$ with $|S|\geq2$, there exists a non-trivial cut $(S_1,S_2)$ such that
$|\edg(S_1,S_2)|\leq\min\{|S_1|,|S_2|\}.$
\end{lem}

\begin{proof}
If $S$ contains no edges, an arbitrary partition will satisfy the requirement. Otherwise, we consider the edges in $S$ and pick an edge with the maximum level. Say this edge is $e=\{\vec{v}_{j_1},\vec{v}_{j_2}\}$, and $\dir(e)=i_0$, $\lev(e)=k$. We assume $\langle\vec{v}_{j_2}-\vec{v}_{j_1},\vec{e}_{i_0}\rangle>0$. Then since this edge is good, $C_k(\vec{v}_{j_1})$ and $C_k(\vec{v}_{j_2})$ are adjacent grid points,
\[C_k(\vec{v}_{j_1})+g_k\vec{e}_{i_0}=C_k(\vec{v}_{j_1}+g_k\vec{e}_{i_0})=C_k(\vec{v}_{j_2}).\]
For any point $\vec{v}\in\R$ and $i\in[d]$, we use $C_k(\vec{v})_i$ to denote the $i$'th coordinate of $C_k(\vec{v})$. Let $h=\left[C_k(\vec{v}_{j_1})_{i_0}+C_k(\vec{v}_{j_2})_{i_0}\right]/2$. We define $S_1$ and $S_2$ as follows.
\begin{eqnarray*}
S_1 & = & \{\vec{v}\in S\mid C_k(\vec{v})_{i_0}<h\}, \\
S_2 & = & \{\vec{v}\in S\mid C_k(\vec{v})_{i_0}>h\}.
\end{eqnarray*}
We can see that $S_1$ and $S_2$ are not empty because $\vec{v}_{j_1}\in S_1$ and $\vec{v}_{j_2}\in S_2$. There is no point $\vec{v}$ satisfying $C_k(\vec{v})_{i_0}=h$, because $C_k(\vec{v})$ is a grid point and $h$ is not a multiple of $g_k$. Hence $(S_1,S_2)$ is a non-trivial cut of $S$.

We consider the edges in $\edg(S_1,S_2)$, and show that every edge in $\edg(S_1,S_2)$ must have direction $i_0$. Assume this is not true, and let $e'=\{\vec{v}_{j_1}',\vec{v}_{j_2}'\}$ be such an edge. Say $\dir(e')=i'$ ($i'\neq i_0$). There are two cases.
\begin{enumerate}
\item $\lev(e')=k$. By the first requirement in the definition of good edges,
\[C_k(\vec{v}_{j_1}')+g_k\vec{e}_{i'}=C_k(\vec{v}_{j_1}'+g_k\vec{e}_{i'})=C_k(\vec{v}_{j_2}').\]
\item $\lev(e')<k$. By the second requirement in the definition of good edges,
\[C_k(\vec{v}_{j_1}')=C_k(\vec{v}_{j_2}').\]
\end{enumerate}
In both cases we have $C_k(\vec{v}_{j_1}')_{i_0}=C_k(\vec{v}_{j_2}')_{i_0}$. Hence the edge $e'\notin\edg(S_1,S_2)$.

Therefore all edges in $\edg(S_1,S_2)$ have direction $i_0$. Since the edges of the same direction are disjoint, we have $|\edg(S_1,S_2)|\leq\min\{|S_1|,|S_2|\}$.
\end{proof}

\section{Lower Bound for $c$-bounded Simple Codes} \label{sec:cbounded}

In this section we prove $n=2^{\Omega(\alpha^2\delta^2d/(\log c)^2)}$ for $c$-bounded $(\alpha,\delta)$ simple codes. The following simple lemma shows that it suffices to consider the $2$-bounded codes.

\begin{lem} \label{lem:bdred}
A $c$-bounded $(\alpha,\delta)$ simple code is  a $2$-bounded $(\alpha,\delta')$ simple code for $\delta'\geq\delta/\lceil\log_2c\rceil$.
\end{lem}

\begin{proof}
We partition the interval $[1,c]$ into $\lceil\log_2c\rceil$ intervals
\[[1,c]=[2^0,2^1)\cup[2^1,2^2)\cup\cdots\cup[2^{\lceil\log_2c\rceil-1},c].\]
By the Pigeonhole Principle, there is an interval that $1/\lceil\log_2c\rceil$ fraction of the edges have lengths in it. We only consider these edges, and scale the points in $V$ to make all these edge lengths in~$[1,2]$. The code becomes $2$-bounded and the density is at least $\delta/\lceil\log_2c\rceil$.
\end{proof}

\subsection{Preliminaries and warm-up}
We let~$\N = \{0,1,2,\dots\}$ and for a vector~$\vec{\sigma}\in\N^d$ we write~$|\vec{\sigma}| = \sigma_1 + \cdots + \sigma_d$.
We denote by~$\E_{\vec{x}\in_{\gamma}\R^d}$ the expectation with respect to a random $d$-dimensional vector $\vec x$ whose entries are independent standard Gaussian random variables.

We collect some basic facts of the Hermite polynomials (see e.g.,~\cite[Section 6.1]{Andrews:1999}).
These polynomials form a complete orthonormal basis for the Hilbert space of square integrable functions~$f:\R^n\to \C$ endowed with the inner product
$(f,g) = \E_{\vec{x}\in_{\gamma}\R^d}\big[ \overline{f(\vec{x})}g(\vec{x})\big]$.
The polynomials can be obtained by Gram-Schmidt orthogonalization on the  monomials $x_1^{\sigma_1}\cdots x_d^{\sigma_d}$ for~$\vec{\sigma}\in\N^d$.
Each Hermite polynomial $h_{\vec{\sigma}}\in\R[x_1,\dots,x_d]$ can thus be uniquely represented by a nonnegative integer vector~$\vec{\sigma}$ and the linear ones are $h_{\vec{e}_i}(\vec{x}) = \langle \vec{e}_i,\vec{x}\rangle = x_i.$
We define the Fourier-Hermite coefficients of a function~$f$ by~$\widehat{f}(\vec{\sigma}) = (h_{\vec{\sigma}}, f)$.
Orthonormality easily gives {\em Parseval's identity}\,:
\begin{equation}\label{eq:parseval}
\sum_{\vec{\sigma}\in\N^d}\widehat f(\vec{\sigma})^2 = \E_{\vec{x}\in_{\gamma}\R^d}\big[|f(x)|^2\big].
\end{equation}

We exploit a connection between particular functions related to the Hermite polynomials and the Fourier transform over~$\R^d$.
Recall that the Fourier transform of a function $f$ at a point~$\vec{y}\in\R^d$ is given by
\begin{equation*}
\big(\mathcal F(f)\big)(\vec{y}) = (2\pi)^{-d/2}\int_{\R^d}f(\vec{x})\, e^{-i\langle \vec{x},\vec{y}\rangle}\, d\vec{x}.
\end{equation*}
The fact we use is that the functions $H_{\vec{\sigma}}(\vec{x}) = e^{-\|\vec{x}\|_2^2/2}\,h_{\vec{\sigma}}(\vec{x})$ (known as the Hermite functions) are eigenfunctions of the Fourier transform: they satisfy~$\mathcal F(H_{\vec{\sigma}}) = (-i)^{|\vec{\sigma}|}H_{\vec{\sigma}}$.
In particular this gives the useful identity
\begin{equation}\label{eq:hermfourier}
\E_{\vec{x}\in_{\gamma}\R^d}\Big[h_{\vec{e}_i}(\vec{x})\, e^{-i\langle \vec{x},\vec{y}\rangle}\Big] =
\big(\mathcal F(H_{\vec{e}_i})\big)(\vec{y})= 
-iH_{\vec{e}_i}(\vec{y})=
\frac{-i\,\langle \vec{e}_i,\vec{y}\rangle}{e^{\|\vec{y}\|_2^2/2}}.
\end{equation}

As a `warm-up' to the $2$-query case, we show how one can prove the following simple bound on 1-query LDCs using these tools.

\begin{lem}\label{lem:1ldc}
Let~$(\vec{v}_1,\dots,\vec{v}_n)\in (\R^d)^n$ be a 1-query $(\alpha, \delta)$-approximate LDC.
Then, $d \leq e/(\alpha^2\delta)$.
\end{lem}

\begin{proof}
Without loss of generality we may assume that the vectors~$\vec{v}_s$, $s\in[n]$, have unit 2-norm.
Define the {\em vector-valued} function~$f:\R^d\to\C^n$ by $f(\vec{x}) = (e^{-i\langle \vec{v}_s, \vec{x}\rangle})_{s=1}^n$.
Define the (vector-valued) Fourier-Hermite coefficients of~$f$ in the obvious way by $\widehat f(\vec{\sigma}) = \E_{\vec{x}\in_{\gamma}\R^d}[h_{\vec{\sigma}}(\vec{x})f(\vec{x})]$.
Parseval's identity~\ref{eq:parseval} applied to the coordinates of~$f$ gives
\begin{equation}\label{eq:1ldc1}
\sum_{i=1}^d \|\widehat f(\vec{e}_i)\|_2^2 \leq \sum_{\vec{\vec{\sigma}}\in\N^d} \|\widehat f(\vec{\sigma})\|_2^2 = \E_{\vec{x}\in_{\gamma}\R^d}\big[\|f(\vec{x})\|^2_2\big].
\end{equation}
The right-hand side of~\eqref{eq:1ldc1} clearly equals~$n$. 
By~\eqref{eq:hermfourier} the left-hand side is at least
\begin{equation}
\sum_{i=1}^d \|\widehat f(\vec{e}_i)\|_2^2 = 
\sum_{i=1}^d \sum_{s=1}^n \big| \E_{\vec{x}\in_{\gamma}\R^d}\big[h_{\vec{e}_i}(\vec{x})e^{-i\langle \vec{v}_s, \vec{x}\rangle}\big] \big|^2 = 
\sum_{i=1}^d \sum_{s=1}^n \frac{|\langle \vec{e}_i, \vec{v}_s\rangle|^2}{e} \geq
\frac{\delta d n \alpha^2}{e},
\end{equation}
where the last inequality follows from the definition of a 1-query approximate LDC.
Putting things together gives $\alpha^2\delta dn/e \leq n$, which implies the result.
\end{proof}

\subsection{A matrix valued function from a 2-query code}

Let $\{\vec{v}_1,\vec{v}_2,\ldots,\vec{v}_n\}$ be a (2-query) $2$-bounded $(\alpha,\delta)$ simple code. 
We define the vector-valued function $f:\R^d\to \C^n$ given by  $f(\vec{x}) = ( e^{-i\langle \vec{x}, \vec{v}_s\rangle})_{s=1}^n$.
And from $f$ we define the matrix-valued function $F(\vec{x}) = f(\vec{x})f(\vec{x})^*$, so
\[F(\vec{x})=\left(e^{-i\langle \vec{x}, \vec{v}_s - \vec{v}_t\rangle}\right)_{s,t=1}^n.\]
Note that each~$F(\vec{x})$ is a Hermitian matrix with rank~1. 
%
We define the matrix-valued Fourier-Hermite coefficients~$\widehat{F}(\vec{\sigma})$ in the obvious way by $\widehat{F}(\vec{\sigma})_{s,t} = ( h_\sigma, F_{s,t})$, where~$F_{s,t}$ is the function corresponding to the $(s,t)$-coordinate of~$F$.


By~\eqref{eq:hermfourier}~$\widehat{F}(\vec{e}_i)$ therefore has as~$(s,t)$-entry given by
\begin{equation*}
\widehat{F}(\vec{e}_i)_{s,t}  
= \E_{x\in_\gamma\R^d} \Big[h_{\vec{e}_i}(\vec{x}) e^{-i\langle \vec{x}, \vec{v}_s - \vec{v}_t\rangle}\, \Big]
\stackrel{\eqref{eq:hermfourier}}{=}  \frac{-i\langle \vec{e}_i, \vec{v}_s - \vec{v}_t\rangle}{e^{\|\vec{v}_s - \vec{v}_t\|_2^2/2}}.
\end{equation*}

Since the code is simple and 2-bounded there are~$|M_i|$ disjoint $\{s,t\}$-pairs such that the~$(s,t)$-entry of~$\widehat F(\vec{e}_i)$ has magnitude
\begin{equation}\label{eq:decodingmatrices}
|\widehat F(\vec{e}_i)_{s,t}| \geq \frac{|\langle \vec{e}_i, \vec{v}_s - \vec{v}_t\rangle|}{e^{\|\vec{v}_s - \vec{v}_t\|_2^2/2}} \geq \frac{\alpha}{e}.
\end{equation}
The matrix~$\widehat F(\vec{e}_i)$ thus has large entries (in absolute values) on the coordinates corresponding to the matching~$M_i$.

\subsection{A Fourier inequality for the trace norm} \label{sec:bdinequality}

We now establish a matrix analog of~\eqref{eq:1ldc1} (Lemma~\ref{lem:traceineq} below), which is expressed in terms of the Schatten-1 norm (or trace norm).
The {\em Schatten-1 norm}~$\|A\|_{S_1}$ of a complex matrix~$A$ is defined as the sum of its singular values.
We also use the following dual characterization of this norm.
For a pair of matrices~$A,X\in\C^{n\times n}$ let~$\langle A,X\rangle = \tr[A^*X]$ denote their trace inner product, where~$A^*$ denotes the conjugate transpose of~$A$.
The {\em spectral norm} $\|X\|_{S_\infty}$ of a matrix~$X$ is defined as its maximum singular value.
We have the  well-known duality characterization
\begin{equation}\label{eq:S1dual}
\|A\|_{S_1} = \max\{|\langle A, X\rangle| :\, \|X\|_{S_\infty}\leq 1\}.
\end{equation}
The analog of~\eqref{eq:1ldc1} that allows us to prove the lower bound on $2$-bounded simple codes is as follows.

\begin{lem}\label{lem:traceineq}
Let $F:\R^d\to\C^{n\times n}$ be a Hermitian matrix-valued function.
Then,
\begin{equation*}
\left(\sum_{i=1}^d \big\| \widehat{F}(\vec{e}_i)\big\|_{S_1}^2  \right)^{1/2}
\leq
\sqrt{2\log (2en)}\left(\E_{\vec x\in_{\gamma} \R^n} \Big[\|F(\vec{x})\|_{S_1}^2 \Big]\right)^{1/2}.
\end{equation*}
\end{lem}

%
%

The proof of this lemma relies on the following non-commutative version of the Khintchine inequality~\cite[Section~4.4]{Tropp:2012}.

\begin{thm}[Non-commutative Khintchine inequality~\cite{Tropp:2012}]\label{thm:nckhintchine}
For any collection of Hermitian matrices~$A_1,\dots,A_d\in\C^{n\times n}$ and i.i.d.\ standard Gaussian random variables~$x_1,\dots,x_d$, we have
\begin{equation*}
\left(\E\left[\Big\|\sum_{i=1}^d x_i\, A_i\Big\|_{S_\infty}^2\right]\right)^{1/2} \leq \sqrt{2\log (2en)}\: \Big\|\sum_{i=1}^d A_i^2\Big\|_{S_\infty}^{1/2}.
\end{equation*}
\end{thm}


\begin{proof}[ of Lemma~\ref{lem:traceineq}]
By homogeneity we may assume that $\|\widehat F(\vec{e}_1)\|_{S_1}^2 + \cdots + \|\widehat F(\vec{e}_d)\|_{S_1}^2 = 1$.
Let~$X_1,\dots,X_d\in\C^{n\times n}$ be such that~$\|X_i\|_{S_\infty}\leq 1$ and $\langle \widehat F(\vec{e}_i), X_i\rangle = \|\widehat F(\vec{e}_i)\|_{S_1}$ for every~$i\in[d]$.
Let~$Y_i = \|\widehat F(\vec{e}_i)\|_{S_1}\, X_i$ and notice that
\begin{equation}\label{eq:Ybound}
\sum_{i=1}^d\|Y_i\|_{S_\infty}^2 \leq \sum_{i=1}^d\|\widehat F(\vec{e}_i)\|_{S_1}^2 = 1.
\end{equation}

We consider the quantity
\begin{equation}\label{eq:bigip}
\E_{\vec x\in_\gamma \R^d}\Big[\Big\langle F(\vec{x}),\sum_{i=1}^d x_i\, Y_i\Big\rangle \Big].
\end{equation}

First, by linearity of the trace function,  it equals
\begin{eqnarray*}
\sum_{i=1}^d\Big\langle \E_{\vec x\in_\gamma \R^d} \big[x_i\, F(\vec{x})\big], Y_i\Big\rangle
= \sum_{i=1}^d \big\langle \widehat F(\vec{e}_i), Y_i\big\rangle
= \sum_{i=1}^d \big\|\widehat F(\vec{e}_i)\big\|_{S_1}^2 = 1.
\end{eqnarray*}

Second, by H\"older's inequality for the trace and spectral norms \cite{Bhatia97} (which follows from~\eqref{eq:S1dual}) 
and the Cauchy-Schwarz inequality, \eqref{eq:bigip} is at most
\begin{multline*}
\E_{\vec x\in_\gamma \R^d}\Big[\|F(\vec{x})\|_{S_1}\Big\|\sum_{i=1}^d x_iY_i\Big\|_{S_\infty} \Big]
\leq
\left(\E_{\vec x\in_\gamma \R^d}\big[ \|F(\vec{x})\|_{S_1}^2\big]\right)^{1/2} \left(\E_{\vec x\in_\gamma \R^d}\Big[ \Big\|\sum_{i=1}^d x_i\,Y_i\Big\|_{S_\infty}^2\Big]\right)^{1/2}.
\end{multline*}
By Theorem~\ref{thm:nckhintchine}, the triangle inequality and the fact~$\|Y_i^2\|_{S_\infty} \leq \|Y_i\|_{S_\infty}^2$, the last factor is at most
\begin{equation*}
\sqrt{2\log (2en)}\: \Big\|\sum_{i=1}^d Y_i^2\Big\|_{S_\infty}^{1/2}
\leq
\sqrt{2\log (2en)}\left(\sum_{i=1}^d\|Y_i\|_{S_\infty}^2 \right)^{1/2} \leq \sqrt{2\log (2en)}.
\end{equation*}
\end{proof}

\subsection{Lower bound on $2$-bounded codes} \label{sec:bdproof}

We now combine fact~\eqref{eq:decodingmatrices} and Lemma~\ref{lem:traceineq} to lower bound the length of 2-bounded simple codes.
Recall that we defined the matrix-valued function $F(\vec{x}) = f(\vec{x})f(\vec{x})^*$ where~$f(\vec{x}) = \left( e^{-i\langle \vec{x}, \vec{v}_s\rangle} \right)_{s=1}^n$.
Then $F(\vec{x})$ has~$f(\vec{x})$ as an eigenvector with eigenvalue~$n$ (its other eigenvalues being zero). Hence, 
\begin{equation}\label{eq:Fxnorm}
\|F(\vec{x})\|_{S_1} = n.
\end{equation}

Recall from~\eqref{eq:decodingmatrices} that for every~$i\in[d]$ there are~$|M_i|$ disjoint $\{s,t\}$-pairs such that $|\widehat F(\vec{e}_i)_{s,t}| \geq \alpha/e$.
From~\cite[p.~14--15]{Ben-Aroya:2008} it directly follows that
\begin{equation}\label{eq:hatbound}
\|\widehat{F}(\vec{e}_i)\|_{S_1} \geq \frac{\alpha}{e} |M_i|.
\end{equation}
Putting the above facts together  gives
\begin{multline*}
\frac{\alpha\delta\sqrt{d}n}{e}
 \leq  
\frac{\alpha}{e}\frac{1}{\sqrt{d}}\sum_{i=1}^d |M_i|
\leq
\left(\frac{\alpha^2}{e^2}\sum_{i=1}^d |M_i|^2\right)^{1/2}
\stackrel{\eqref{eq:hatbound}}{\leq}
\left(\sum_{i=1}^d\|\widehat{F}(\vec{e}_i)\|_{S_1}^2\right)^{1/2}\\[.2cm]
\stackrel{\text{Lemma}~\ref{lem:traceineq}}{\leq}
\sqrt{2\log (2en)}\, \left(\E_{\vec x\in_\gamma \R^d}\big[\|F(\vec{x})\|_{S_1}^2\big]\right)^{1/2}
\stackrel{\eqref{eq:Fxnorm}}{=}
\sqrt{2\log (2en)}\, n,
\end{multline*}
where the second inequality follows from Cauchy-Schwarz.
Hence, $n=2^{\Omega(\alpha^2\delta^2d)}$.

\section{Approximate $q$-query Code for General $q$} \label{sec:qquery}

In this section we prove Theorem~\ref{thm:qquery} by showing that $n=\Omega((\alpha^2\delta^{1/q}d)^{\frac{q}{q-1}})$ for general $q$-query approximate code. The proof is similar to~\cite{KT00}: we select a random subset of $V$ with size $\Theta(\delta^{-\frac{1}{q}}n^{\frac{q-1}{q}})$, and show that w.h.p it contains a $q$-tuple from at least $\Omega(d)$ matchings (or directions). This will imply the size of the subset is $\Omega(\alpha^2d)$. The lower bound of $n$ follows immediately.

We first note that a subset containing tuples from many different matchings must be large.

\begin{lem} \label{lem:kt}
If a set $S\subseteq V$ contains at least one tuple from $k$ different matchings ($k\leq d$), then $|S|\geq\alpha^2 k$.
\end{lem}
\begin{proof}
By the definition of an approximate LDC, for every tuple $\{{j_1},\ldots,{j_k}\}\in M_i$, there exists a unit vector $\vec{u}\in\spn\{\vec{v}_{j_1},\ldots,\vec{v}_{j_k}\}$ with the $i$'th coordinate at least $\alpha$ in absolute value. We also assume w.l.o.g. that $u_i\geq\alpha$ (otherwise take $-\vec{u}$). Therefore, there exists unit vectors $\vec{u}_1,\vec{u}_2,\ldots,\vec{u}_k\in\spn\{S\}$ such that each of them has a different coordinate at least $\alpha$. Without loss of generality we assume $u_{11},u_{22},\ldots,u_{kk}\geq\alpha$.

To show $|S|\geq\alpha^2 k$, it suffices to show $\rank\{\vec{u}_1,\vec{u}_2,\ldots,\vec{u}_k\}\geq\alpha^2 k$. Let $U$ be the matrix consisting of $\vec{u}_1,\vec{u}_2,\ldots,\vec{u}_k$ as its row vectors. For simplicity we remove the last $n-k$ columns in $U$ if $n>k$. Now $U$ is a square matrix with diagonal elements at least $\alpha$. Let $r$ be the rank of $U$. We only need to show $r\geq\alpha^2 k$. This is a variant of the well-known theorem saying that the rank of a matrix is large if its diagonal elements are larger than the off-diagonal ones. We give the following proof, which is based on the idea of Lemma~3.5 in~\cite{BDWY12} (similar lemmas can be found in most standard texts on matrix analysis).

Let $U=Q\Sigma P^*$ be the singular value decomposition of $U$, where $Q,P$ are unitary matrices. Let $\sigma_1,\sigma_2,\ldots,\sigma_r>0$ be the non-zero singular values of $U$, i.e. $\Sigma=\diag\{\sigma_1,\sigma_2,\ldots,\sigma_r,0,0,\ldots,0\}$. We have
\[(\alpha k)^2=\tr(U)^2=\tr(Q\Sigma P^*)^2=\tr((P^*Q)\Sigma)^2\leq\tr(\Sigma)^2.\]
The last inequality holds since $P^*Q$ is a unitary matrix and every element has absolute value at most $1$. On the other hand,
\[\tr(\Sigma)^2=(\sigma_1+\sigma_2+\cdots+\sigma_r)^2\leq r\cdot(\sigma_1^2+\sigma_2^2+\cdots+\sigma_r)^2=r\cdot\|U\|_F\leq r\cdot k.\]
Combine these two inequalities we have $r\geq\alpha^2 k$.
\end{proof}

Now we can prove the theorem in the same way as in~\cite{KT00}. We pick a random set $S\subseteq V$ of size $\Theta(\delta^{-\frac{1}{q}}n^{\frac{q-1}{q}})$. By Lemma~5 in~\cite{KT00}, $S$ contains tuples in $\Omega(d)$ different directions in expectation. We fix an $S$ that contains tuples in $\Omega(d)$ directions. By Lemma~\ref{lem:kt}, we have
\[\delta^{-\frac{1}{q}}n^{\frac{q-1}{q}}=\Omega(\alpha^2d).\]
The lower bound $n=\Omega((\alpha^2\delta^{1/q}d)^{\frac{q}{q-1}})$ follows immediately.


\bibliographystyle{alpha}

\bibliography{ApproxLDCs}

\newcommand{\etalchar}[1]{$^{#1}$}
\begin{thebibliography}{BARdW08}

\bibitem[AAR99]{Andrews:1999}
G.~E. Andrews, R.~Askey, and R.~Roy.
\newblock {\em Special Functions}, volume~71 of {\em Encyclopedia of
  Mathematics and its Applications}.
\newblock Cambridge University Press, 1999.

\bibitem[ADSW12]{ADSW12}
A.~Ai, Z.~Dvir, S.~Saraf, and A.~Wigderson.
\newblock {{S}ylvester-{G}allai theorems for approximate collinearity}.
\newblock Forum of mathematics - Sigma (to appear), 2012.

\bibitem[BARdW08]{Ben-Aroya:2008}
A.~Ben-Aroya, O.~Regev, and R.~de~Wolf.
\newblock A hypercontractive inequality for matrix-valued functions with
  applications to quantum computing and {LDC}s.
\newblock In {\em FOCS'08}, pages 477--486, 2008.

\bibitem[BDWY12]{BDWY12}
B.~Barak, Z.~Dvir, A.~Wigderson, and A.~Yehudayoff.
\newblock {Fractional {S}ylvester-{G}allai theorems}.
\newblock {\em Proceedings of the National Academy of Sciences}, 2012.

\bibitem[BET10]{BET10}
A.~Ben{-}Aroya, K.~Efremenko, and A.~Ta{-}Shma.
\newblock Local list decoding with a constant number of queries.
\newblock In {\em FOCS'10}, pages 715--722, 2010.

\bibitem[BF90]{BeaverF90}
D.~Beaver and J.~Feigenbaum.
\newblock Hiding instances in multioracle queries.
\newblock In {\em STACS'90}, pages 37--48, 1990.

\bibitem[Bha97]{Bhatia97}
R.~Bhatia.
\newblock {\em Matrix Analysis}, volume 169 of {\em Graduate Texts in
  Mathematics}.
\newblock Springer, 1997.

\bibitem[BK95]{Blum-Kannan}
M.~Blum and S.~Kannan.
\newblock Designing programs that check their work.
\newblock {\em J. ACM}, 42(1):269--291, January 1995.

\bibitem[BNR12]{Briet:2012}
J.~Bri\"{e}t, A.~Naor, and O.~Regev.
\newblock Locally decodable codes and the failure of cotype for projective
  tensor products.
\newblock {\em Electronic Research Announcements in Mathematical Sciences
  (ERA-MS)}, 19:120--130, 2012.

\bibitem[CFL{\etalchar{+}}10]{CFL+10}
Y.~M. Chee, T.~Feng, S.~Ling, H.~Wang, and L.~F. Zhang.
\newblock Query-efficient locally decodable codes of subexponential length.
\newblock {\em Electronic Colloquium on Computational Complexity (ECCC)},
  TR10-173, 2010.

\bibitem[DGY11]{DGY11}
Z.~Dvir, P.~Gopalan, and S.~Yekhanin.
\newblock Matching vector codes.
\newblock {\em SIAM J. Comput.}, 40(4):1154--1178, 2011.

\bibitem[DS05]{DS05}
Z.~Dvir and A.~Shpilka.
\newblock Locally decodable codes with 2 queries and polynomial identity
  testing for depth 3 circuits.
\newblock In {\em STOC'05}, pages 592--601, 2005.

\bibitem[DSW12]{DSW12}
Z.~Dvir, S.~Saraf, and A.~Wigderson.
\newblock {Improved rank bounds for design matrices and a new proof of
  {K}elly's theorem.}
\newblock Forum of mathematics - Sigma (to appear), 2012.

\bibitem[DSW13]{DSW13}
Z.~Dvir, S.~Saraf, and A.~Wigderson.
\newblock {Breaking the quadratic barrier for 3-{LCC}s over the reals}.
\newblock Manuscript, 2013.

\bibitem[Efr09]{Efr09}
K.~Efremenko.
\newblock 3-query locally decodable codes of subexponential length.
\newblock In {\em STOC'09}, pages 39--44, 2009.

\bibitem[GKST06]{GKST:2006}
O.~Goldreich, H.~Karloff, L.~J. Schulman, and L.~Trevisan.
\newblock Lower bounds for linear locally decodable codes and private
  information retrieval.
\newblock {\em Computational Complexity}, 15(3):263--296, 2006.

\bibitem[IS10]{IS10}
T.~Itoh and Y.~Suzuki.
\newblock Improved constructions for query-efficient locally decodable codes of
  subexponential length.
\newblock {\em IEICE Transactions on Information and Systems},
  E93-D(2):263--270, 2010.

\bibitem[KdW04]{KdW04}
I.~Kerenidis and R.~de~Wolf.
\newblock Exponential lower bound for 2-query locally decodable codes via a
  quantum argument.
\newblock {\em Journal of Computer and System Sciences}, 69(3):395--420, 2004.

\bibitem[KROW12]{KORW12}
G.~Kindler, A.~Rao, R.~O'Donnell, and A.~Wigdersons.
\newblock Spherical cubes: optimal foams from computational hardness
  amplification.
\newblock {\em Commun. ACM}, 55(10):90--97, October 2012.

\bibitem[KT00]{KT00}
J.~Katz and L.~Trevisan.
\newblock On the efficiency of local decoding procedures for error-correcting
  codes.
\newblock In {\em STOC'00}, pages 80--86, New York, NY, USA, 2000. ACM.

\bibitem[KY09]{KY09}
K.~S. Kedlaya and S.~Yekhanin.
\newblock Locally decodable codes from nice subsets of finite fields and prime
  factors of {M}ersenne numbers.
\newblock {\em SIAM J. Comput.}, 38(5):1952--1969, 2009.

\bibitem[Lip90]{Lipton90}
R.~J. Lipton.
\newblock Efficient checking of computations.
\newblock In {\em STACS'90}, pages 207--215, 1990.

\bibitem[Pra07]{Rag07}
R.~Prasad.
\newblock A note on {Y}ekhanin's locally decodable codes.
\newblock {\em Electronic Colloquium on Computational Complexity (ECCC)},
  TR07-016, 2007.

\bibitem[Tro12]{Tropp:2012}
J.~A. Tropp.
\newblock User-friendly tail bounds for sums of random matrices.
\newblock {\em Foundations of Computational Mathematics}, 12(4):389--434, 2012.

\bibitem[Woo07]{Woo07}
D.~P. Woodruff.
\newblock New lower bounds for general locally decodable codes.
\newblock {\em Electronic Colloquium on Computational Complexity (ECCC)},
  TR07-006, 2007.

\bibitem[Woo12]{Woo10}
D.~P. Woodruff.
\newblock A quadratic lower bound for three-query linear locally decodable
  codes over any field.
\newblock {\em Journal of Computer Science and Technology}, 27(4):678--686,
  2012.

\bibitem[Yek08]{Yek08}
S.~Yekhanin.
\newblock Towards 3-query locally decodable codes of subexponential length.
\newblock {\em Journal of the ACM}, 55(1):1--16, 2008.

\end{thebibliography}

\end{document}